\def\UseBibLatex{1}
\ifx\SoCG\undefined%
\documentclass[11pt]{article}%
\providecommand{\SoCGVer}[1]{}%
\providecommand{\NotSoCGVer}[1]{#1}%
\else%
\makeatletter
\def\input@path{{lipics/}{../lipics/}}
\makeatother
\documentclass[a4paper,USenglish,cleveref,autoref,thm-restate]%
{socg-lipics-v2021}

\hideLIPIcs 
\providecommand{\SoCGVer}[1]{#1}%
\providecommand{\NotSoCGVer}[1]{}%
\fi

\NotSoCGVer{%
   \usepackage[in]{fullpage}%
}%
\usepackage{graphicx}%
\usepackage{amsmath, amssymb}%
\usepackage{wasysym}%
\usepackage{bm}%
\usepackage{mleftright}%
\usepackage[dvipsnames]{xcolor}%
\usepackage{xspace}%
\usepackage{xparse}
\usepackage{caption}
\usepackage{stmaryrd}%
\usepackage[ruled,vlined,linesnumbered]{algorithm2e}%
\usepackage{commath}%
\usepackage{scalerel}%

\usepackage[utf8]{inputenc}

\NotSoCGVer{%
   \usepackage[amsmath,thmmarks]{ntheorem}%
   \theoremseparator{.}%
}
\usepackage{titlesec}%
\titlelabel{\thetitle. }%

\titleformat{\paragraph}[runin]
  {\normalfont\bfseries}
  {\theparagraph}
  {1em}
  {\addperiod}

\newcommand{\addperiod}[1]{#1.}

\DeclareFontFamily{U}{BOONDOX-calo}{\skewchar\font=45 }
\DeclareFontShape{U}{BOONDOX-calo}{m}{n}{
  <-> s*[1.05] BOONDOX-r-calo}{}
\DeclareFontShape{U}{BOONDOX-calo}{b}{n}{
  <-> s*[1.05] BOONDOX-b-calo}{}
\DeclareMathAlphabet{\mathcalb}{U}{BOONDOX-calo}{m}{n}
\SetMathAlphabet{\mathcalb}{bold}{U}{BOONDOX-calo}{b}{n}
\DeclareMathAlphabet{\mathbcalb}{U}{BOONDOX-calo}{b}{n}

\usepackage[inline]{enumitem}

\newlist{compactenumA}{enumerate}{5}%
\setlist[compactenumA]{topsep=0pt,itemsep=-1ex,partopsep=1ex,parsep=1ex,%
   label=(\Alph*)}%

\newlist{compactenuma}{enumerate}{5}%
\setlist[compactenuma]{topsep=0pt,itemsep=-1ex,partopsep=1ex,parsep=1ex,%
   label=(\alph*)}%

\newlist{compactenumI}{enumerate}{5}%
\setlist[compactenumI]{topsep=0pt,itemsep=-1ex,partopsep=1ex,parsep=1ex,%
   label=(\Roman*)}%

\newlist{compactenumi}{enumerate}{5}%
\setlist[compactenumi]{topsep=0pt,itemsep=-1ex,partopsep=1ex,parsep=1ex,%
   label=(\roman*)}%

\newlist{compactitem}{itemize}{5}%
\setlist[compactitem]{topsep=0pt,itemsep=-1ex,partopsep=1ex,parsep=1ex,%
   label=\bullet}%

\usepackage{hyperref}%
\hypersetup{%
   breaklinks,%
   colorlinks=true,%
   urlcolor=[rgb]{0.25,0.0,0.0},%
   linkcolor=[rgb]{0.5,0.0,0.0},%
   citecolor=[rgb]{0,0.2,0.445},%
   filecolor=[rgb]{0,0,0.4}
}
\numberwithin{figure}{section}%
\numberwithin{table}{section}%
\numberwithin{equation}{section}%

\definecolor{blue25emph}{rgb}{0, 0, 11}
\newcommand{\emphic}[2]{%
   \textcolor{blue25emph}{%
      \textbf{\emph{#1}}}%
   \index{#2}}

\newcommand{\emphi}[1]{\emphic{#1}{#1}}

\definecolor{almostblack}{rgb}{0, 0, 0.3}
\providecommand{\emphw}[1]{{\textcolor{almostblack}{\emph{#1}}}}%

\definecolor{DGray}{RGB}{94,94,94}

\newcommand{\myqedsymbol}{\rule{2mm}{2mm}}
\newcommand{\eps}{\varepsilon}%

\newcommand{\Null}[1]{\mathsf{null}}%

\NotSoCGVer{%
\theoremstyle{plain}%
\newtheorem{theorem}{Theorem}[section]
\newtheorem{FakeCounter}{FakeCounter}%

\newtheorem{lemma}[theorem]{Lemma}

\newtheorem{corollary}[theorem]{Corollary}
 \newtheorem{fact}[theorem]{Fact}

\newtheorem{invariant}[theorem]{Invariant}
\newtheorem{question}[theorem]{Question}

\newtheorem{prop}[theorem]{Proposition}
\newtheorem{openproblem}[theorem]{Open Problem}
\newtheorem{problem}[theorem]{Problem}

\theoremstyle{plain}%
\theoremheaderfont{\sf}%
\theorembodyfont{\upshape}%

\newtheorem*{defn:unnumbered}[FakeCounter]{Definition}
\newtheorem*{remark:unnumbered}[FakeCounter]{Remark}%

\newtheorem{definition}[theorem]{Definition}
\newtheorem{defn}[theorem]{Definition}

\newtheorem{xca}[theorem]{Exercise}
\newtheorem{exercise_h}[theorem]{Exercise}
\newtheorem{assumption}[theorem]{Assumption}%

\theoremheaderfont{\em}%
\theorembodyfont{\upshape}%
\theoremstyle{nonumberplain}%
\theoremseparator{}%
\theoremsymbol{\myqedsymbol}%
\newtheorem{proof}{Proof.}%
}

\SoCGVer{%
   \theoremstyle{plain}%

   \theoremstyle{plain}%
   \newtheorem{defn}[theorem]{Definition}
}%

\newcommand{\atgen}{\symbol{'100}}
\newcommand{\SarielThanks}[1]{\thanks{Department of Computer Science;
      University of Illinois; 201 N. Goodwin Avenue; Urbana, IL,
      61801, USA; {\tt sariel\atgen{}illinois.edu}; {\tt
         \url{http://sarielhp.org/}.} #1}}

\newcommand{\EWR}[1]{\thanks{Department of Computer Science;
		University of Illinois; 201 N. Goodwin Avenue; Urbana, IL,
		61801, USA; {\tt erobson2\atgen{}illinois.edu}; {\tt
			\url{https://eliotwrobson.github.io/}.} #1}}

\newcommand{\HLinkShort}[2]{\hyperref[#2]{#1\ref*{#2}}}
\newcommand{\HLink}[2]{\hyperref[#2]{#1~\ref*{#2}}}
\newcommand{\HLinkPage}[2]{\hyperref[#2]{#1~\ref*{#2}%
      $_\text{p\pageref{#2}}$}}
\newcommand{\HLinkPageOnly}[1]{\hyperref[#1]{Page~\refpage*{#1}%
      $_\text{p\pageref{#1}}$}}

\newcommand{\HLinkSuffix}[3]{\hyperref[#2]{#1\ref*{#2}{#3}}}
\newcommand{\HLinkPageSuffix}[3]{\hyperref[#2]{#1\ref*{#2}%
      #3$_\text{p\pageref{#2}}$}}

\newcommand{\seclab}[1]{\label{sec:#1}}
\providecommand{\secref}[1]{\HLink{Section}{sec:#1}}
\renewcommand{\secref}[1]{\HLink{Section}{sec:#1}}

\newcommand{\corlab}[1]{\label{cor:#1}}
\newcommand{\corref}[1]{\HLink{Corollary}{cor:#1}}%

\providecommand{\deflab}[1]{\label{def:#1}}
\newcommand{\defref}[1]{\HLink{Definition}{def:#1}}

\newcommand{\defrefY}[2]{\hyperref[def:#1]{#2}}

\newcommand{\lemlab}[1]{\label{lemma:#1}}
\providecommand{\lemref}[1]{\HLink{Lemma}{lemma:#1}}%
\renewcommand{\lemref}[1]{\HLink{Lemma}{lemma:#1}}%

\providecommand{\eqlab}[1]{}%
\renewcommand{\eqlab}[1]{\label{equation:#1}}
\newcommand{\Eqref}[1]{\HLinkSuffix{Eq.~(}{equation:#1}{)}}

\newcommand{\Prob}[1]{\mathop{\ProbLTR} \mleft[ #1 \mright]}%

\newcommand{\disk}{\Mh{\Circle}}%
\newcommand{\pth}[2][\!]{\mleft({#2}\mright)}%
\newcommand{\ceil}[1]{\left\lceil {#1} \right\rceil}

\newcommand{\ProbCond}[2]{\mathop{\ProbLTR}\!\left[%
       #1 \;\middle\vert\; #2 \right]}

\newcommand{\ExChar}{\mathbb{E}}%
\newcommand{\ExSym}{\mathop{\ExChar}}%
\newcommand{\Ex}[2][\!]{\ExSym#1\pbrcx{#2}}

\newcommand{\ProbLTR}{\mathbb{P}}%

\newcommand{\Set}[2]{\left\{ #1 \;\middle\vert\; #2 \right\}}

\newcommand{\Disks}{\Mh{\mathcal{C}}}%
\newcommand{\DS}{\Disks}%

\renewcommand{\Re}{\mathbb{R}}%
\renewcommand{\th}{th\xspace}

\renewcommand{\d}{\operatorname{\Mh{\mathsf{d}}}}
\newcommand{\depthX}[1]{\d \pth{#1}}
\newcommand{\depthY}[2]{\d \pth{#1,#2}}
\providecommand{\Mh}[1]{{#1}}%

\IfFileExists{.latex_printer_friendly}{\def\GenPrinterVer{1}}{}%

\ifx\GenPrinterVer\undefined
   \IfFileExists{.latex_color}{\def\GenColorMath{1}}{}
\else
\fi

\ifx\GenColorMath\undefined
\else
\renewcommand{\Mh}[1]{{\textcolor{red}{#1}}}%
\fi
\providecommand{\BibLatexMode}[1]{}
\providecommand{\BibTexMode}[1]{#1}

\ifx\UseBibLatex\undefined%
  \renewcommand{\BibLatexMode}[1]{}
  \renewcommand{\BibTexMode}[1]{#1}
\else
  \renewcommand{\BibLatexMode}[1]{#1}
  \renewcommand{\BibTexMode}[1]{}
\fi

\newcommand{\UsePackage}[1]{%
  \IfFileExists{../styles/#1.sty}{%
      \usepackage{../styles/#1}%
   }{%
      \IfFileExists{./styles/#1.sty}{%
         \usepackage{styles/#1}%
      }{%
         \usepackage{#1}%
      }%
   }%
}

\BibLatexMode{%
   \usepackage{sariel_biblatex}%
   \usepackage[english]{babel}%
   \usepackage{csquotes}
   \renewcommand*{\multicitedelim}{\addcomma\space}%
}

\newcommand{\pa}{\Mh{p}}%
\newcommand{\pb}{\Mh{q}}%
\newcommand{\cardin}{\abs}%

\providecommand{\G}{\Mh{G}}%

\renewcommand{\G}{\Mh{G}}%
\newcommand{\GA}{\Mh{H}}%

\newcommand{\GD}{\G_{\DS}}

\renewcommand{\H}{\Mh{H}}%
\renewcommand{\S}{\Mh{S}}%

\newcommand{\IRX}[1]{\left\llbracket #1 \right\rrbracket}

\newcommand{\GIY}[2]{\G_{#1} \pth{#2}}
\newcommand{\Arr}{\Mh{\mathcal{A}}}%
\newcommand{\ArrX}[1]{\Arr\pth{#1}}%
\newcommand{\ArrY}[2]{\Arr_{#1}\pth{#2}}%
\newcommand{\pbrcx}[1]{\left[ {#1} \right]}%
\newcommand{\edge}{\Mh{\mathsf{e}}}%
\newcommand{\EdgesX}[1]{\Mh{E}\pth{#1}}%
\newcommand{\EdgeSet}{\Mh{\mathsf{E}}}%
\newcommand{\DSX}[1]{\Mh{\mathcal{F}}_{#1}}
\newcommand{\SZone}{\Mh{\mathcal{Z}}}%
\newcommand{\SZoneY}[3]{\Mh{\mathcal{Z}}_{#1} \pth{#2 - #3}}%
\newcommand{\GS}{\Mh{H}}%
\newcommand{\BSet}{\Mh{B}}%

\providecommand{\remove}[1]{}%

\newcommand{\face}{\Mh{f}}

\newcommand{\VV}{\Mh{V}}%
\newcommand{\VX}[1]{\VV\pth{#1}}%
\newcommand{\EG}{\Mh{E}}%
\newcommand{\EGX}[1]{\Mh{E}\pth{#1}}%

\newcommand{\constA}{\mathsf{c}_\alpha}

\newcommand{\cExpander}{\mathsf{c_e}}%

\newcommand{\nc}{\Mh{\xi}}
\newcommand{\tldO}{\scalerel*{\widetilde{O}}{j^2}}%
\newcommand{\nbrX}[1]{\Mh{N}\pth{#1}}
\newcommand{\nbrY}[2]{\Mh{N}_{#1}\pth{#2}}
\newcommand{\spanner}{\Mh{\mathcal{S}}}%
\newcommand{\spX}[1]{\spanner\pth{#1}}%

\newcommand{\IGC}{\G}
\newcommand{\IGX}[1]{\IGC\pth{#1}}
\newcommand{\n}{\Mh{\nu}}%
\newcommand{\cExpSize}{\mathsf{c}_{s}}%
\newcommand{\cExpRep}{\mathsf{c}_{r}} %
\newcommand{\CC}{\Mh{\mathcal{F}}}
\newcommand{\R}{\Mh{\mathcal{R}}}%

\newcommand{\Leftover}{\Mh{L}}%
\newcommand{\nL}{\Mh{\ell}}%

\BibLatexMode{\bibliography{reliable_disks}}%

\title{Sparsifying Disk Intersection Graphs for Reliable
   Connectivity}%

\NotSoCGVer{%
   \author{%
      Sariel Har-Peled%
      \SarielThanks{Work on this paper was partially supported by NSF AF
         award CCF-1421231.%
      }%
      \and%
      Eliot Wong Robson%
      \EWR{}%
   }%
}

\SoCGVer{%
   \author{Eliot Wong Robson}%
   {Department of Computer Science, University of Illinois, 201
      N. Goodwin Avenue, Urbana, IL 61801, USA}%
   {erobson2@illinois.edu}%
   {https://orcid.org/0000-0002-1476-6715}%
   {}%
   \author{Sariel Har-Peled}%
   {Department of Computer Science, University of Illinois, 201
      N. Goodwin Avenue, Urbana, IL 61801, USA}%
   {sariel@illinois.edu}%
   {https://orcid.org/0000-0003-2638-9635}%
   {Work on this paper was partially supported by NSF AF award
      CCF-1907400.}%
}%

\SoCGVer{%
   \authorrunning{S. Har-Peled and E. W. Robson} %
   \Copyright{Sariel Har-Peled and Eliot Wong Robson}%
   \ccsdesc[500]{Theory of computation~Computational geometry}%
   \keywords{Spanners, intersection graphs, reliability, expanders}%
}

\begin{document}

\maketitle

\begin{abstract}
    The intersection graph induced by a set $\Disks$ of $n$ disks can
    be dense. It is thus natural to try and sparsify them, while
    preserving connectivity. Unfortunately, sparse graphs can always
    be made disconnected by removing a small number of vertices. In
    this work, we present a sparsification algorithm that maintains
    connectivity between two disks in the computed graph, if the
    original graph remains ``well-connected'' even after removing an
    arbitrary ``attack'' set $\BSet \subseteq \Disks$ from both
    graphs. Thus, the new sparse graph has similar reliability to
    the original disk graph, and can withstand
    catastrophic failure of nodes while still providing a
    connectivity guarantee for the remaining graph. The new graphs has
    near linear complexity, and can be constructed in near linear time.

    The algorithm extends to any collection of shapes in the plane,
    such that their union complexity is near linear.
\end{abstract}

\section{Introduction}

Given a set $\DS$ of $n$ disks in the plane, their intersection graph
is formed by connecting each pair of intersecting disks by an
edge. While this graph has an implicit representation of linear size,
its explicit graph representation might be of quadratic size. It is
thus natural to try and replace this graph by a sparser graph that retains
some desired properties, such as preserving distances (i.e., a
spanner), or preserving connectivity.

Such questions becomes significantly more challenging if one want to
preserve such properties under network failures.  The main obstacle is
that a sparse graph can always be made disconnected by deleting the
neighbors of a low degree vertex. Thus, the minimum degree of a graph
has to be high, failing to provide the desired sparsity, especially if
the graph has to withstand large attacks. An alternative approach is
to provide such a guarantee only to most of the remaining graph (after
the failure), allowing some parts of the graph to be ignored. For geometric 
spanners there has been recent work on constructing such reliable spanners
\cite{bho-sda-20,bho-srsal-20,hmo-rsms-21}.

Here we consider the case of geometric intersection graphs, and
connectivity guarantees.  Specifically, given a set of disks $\Disks$ the corresponding intersection graph $\IGX{\Disks}$, our goal is to compute a
sparse subgraph $\GA \subseteq \G$, such that its connectivity is
robust to vertex deletion. However, disk intersection graphs might
have cut vertices -- that is, a single vertex whose removal
disconnects the graph. So even a small number of disks being removed can
dramatically effect the connectivity.

As such, a milder desired property is that for any attack set
$\BSet \subseteq \Disks$, the graphs $\G - \BSet$ and $\GA - \BSet$
(that is, the graphs remaining after the vertices of $\BSet$ are
deleted) have similar connectivity. Even that is not achievable if
the graph $\GA$ is sparse, as one can delete the neighbors of a disk
in the sparse graph, which would leave it isolated in $\GA$ (while
still connected to the remaining graph in $\G$).

\paragraph*{$\eps$-safe paths}

Instead, we seek to provide a more geometric guarantee -- a point $\pa$ is
\emphw{$\eps$-safe} if the attack set $\BSet$ removes, at most, a
$(1-\eps)$-fraction of the disks covering $\pa$. A curve is thus $\eps$-safe
if all the points along it are $\eps$-safe, and two disks are \emphw{$\eps$-safely
   connected} if there is an $\eps$-safe path connecting them. Our
goal here is to pre-build a sparse graph $\GA$ that guarantees
connectivity, for any attack set $\BSet$, in the graph $\GA - \BSet$,
for any two disks that are $\eps$-safely connected in $\G -
\BSet$. Importantly, the graph $\GA$ is constructed before $\BSet$ is
known, and the required property should hold for any attack set
$\BSet$. A graph $\GA$ with this property is \emphw{$\eps$-safely
   connected}.

\subsection*{Our result}

Given a set $\Disks$ of $n$ disks in the plane, and a parameter
$\eps \in (0,1)$, let $\G = \IGX{\Disks}$ denote the intersection graph induced by
$\Disks$. We present a near linear time construction of a sparse
subgraph $\GA$ of $\G$, that is $\eps$-safely connected.

\paragraph*{Idea} %
For a point $\pa \in \Re^2$, let $\GA_\pa$ denote the induced subgraph
of $\GA$ over the set of disks
\begin{equation*}
    \Disks \sqcap \pa = \Set{ \disk \in \Disks}{ \pa \in \disk}.
\end{equation*}
For any point $\pa$ in the plane, the graph $\G_\pa$ is a clique. Our
construction replaces this clique by the graph $\GA_\pa$, which is a
``strong'' expander. This property by itself is sufficient to
guarantee that $\GA$ is safely connected. The challenge is how to
construct $\GA$ such that it has the desired expander property
\emph{for all points in the plane}, while being sparse, and
furthermore do this in near linear time (as building the graph
explicitly takes quadratic time).

\paragraph*{Random coloring, sparsification, and expansion}

It is well known that random coloring of vertices can be used to
sparsify a graph, by keeping only edges that connect certain
colors. For example, if you randomly color a clique of $k$ vertices by
$2k$ colors, and keep only edges that connect vertices that belong to
consecutive colors, then the resulting graph has (in expectation)
$\Theta(k)$ edges. This collection of edges is almost a matching.  It
is well known that the union of three random matchings form an
expander with high probability \cite{p-otcoac-73}. Thus, if one
repeats the random coloring idea suggested above a sufficient number of
times, the union of the collection of edges results in an expander.

\paragraph*{Stop in the shallow parts, before getting too deep}
The problem is that if we randomly color the given set $\Disks$, by
$k$ colors, the deepest point $\pa$ in $\ArrX{\Disks}$ might be of
depth $n$. As such, a single random coloring would replace $\G_\pa$ by
a graph that has $O(n^2/k)$ edges, which is way too many edges to be
used in a sparse graph.

To avoid this problem, we only add edges of the coloring if they
correspond to shallow regions (i.e., regions of depth $\approx k$).
The Clarkson-Shor technique readily implies that the number of edges
added by this is roughly $\tldO(n)$, where $\tldO$ hides polynomial
terms in $1/\eps$ and $\log n$. Repeating this sufficient number of
times (i.e., polylogarithmic), provides the desired property for faces
that are of depth in the range $k$ to $2k$. Repeating this for
exponential scales, to cover all faces of the arrangement by good
``depth'' expanders.

\paragraph*{Generalizing to other families of objects}
The key property of disk graphs which we use to analyze our
algorithm's runtime and size of the resulting spanner is the bound on
the union complexity \cite{cs-arscg-89}. As a result, our techniques
immediately generalize to intersection graphs of other families of
objects with near linear bounds on their union complexity -- for
example, the union complexity of $n$ fat triangles is $O(n \log^* n)$
\cite{abes-ibulf-14}, and our construction would work verbatim for
this case.

\bigskip Beyond the result itself, we believe our combination of
techniques from traditional computational geometry and expanders is
quite interesting, and should be useful for other problems.

\section{Settings}

\subsection{Notations}

For a positive integer $k$, let $\IRX{k} = \{ 1,\ldots, k\}$.

For a
graph $\G = (\Disks, \EG)$, for a set $X \subseteq \Disks$, we denote
by $\G_X = \bigl(X,\Set{uv \in \EG}{u,v \in X}\bigr)$ the
\emphi{induced subgraph} of $\G$ over $X$. For a set
$Y \subseteq \Disks$, let $\G - Y = \IGX{\Disks - Y}$ denote
the graph remaining from $\G$ after deleting all the vertices
of $Y$. Similarly, for $y \in \Disks$, we use the shorthand
$\G-y = \G - \set{y}$.

For a set of vertices $S \subseteq \VX{\G}$, let
$\nbrX{S} = \Set{x \in \VV}{ y \in S, xy \in \EGX{\G}}$ be the
\emphi{neighborhood} of $S$.

\subsubsection{Intersection graph}

For a set of regions $\Disks$  in the plane, let
\begin{equation*}
    \IGC = 
    \IGX{\Disks}
    =%
    \pth{\Disks,
    \Set{\disk_1 \disk_2}{ \disk_1 \cap \disk_2 \neq \varnothing,
       \; \disk_1, \disk_2 \in \Disks}
  	}
\end{equation*}
denote the \emphi{intersection graph} of $\Disks$. Throughout this paper,
we assume that the regions are in general position. For a point $\pa$
in the plane, let
\begin{equation*}
    \Disks \sqcap \pa = \Set{ \disk \in \Disks}{ \pa \in \disk}
\end{equation*}
be the set of disks of $\Disks$ covering $\pa$. The induced subgraph
is denoted by
\begin{math}
    \IGC_{\pa} = \IGC_{\Disks \sqcap \pa}.
\end{math}

\subsection{Problem statement}
Given a set of disks $\Disks$ in the plane, consider the induced
intersection graph $\G = \IGX{\Disks} = (\Disks, \EG)$, where
\begin{math}
    \EG = \Set{\disk_1 \disk_2}{\disk_1 \cap \disk_2 \neq \varnothing,
       \; \disk_1, \disk_2 \in \Disks}.
\end{math}
We consider some arbitrary (unknown) \emphi{attack set}
$\BSet \subset \Disks$ -- the disks in this set are being deleted, and
we are interested in the connectivity of the remaining graph
$\IGX{\Disks - \BSet}$ (that is, the induced
subgraph of $\G$ over $\Disks - \BSet$).

\begin{definition}
    \deflab{depth}%
    For a point $\pa \in \Re^2$, its \emphi{depth} is
    \begin{math}
        \depthX{\pa}%
        =%
        \depthY{\pa}{\Disks}%
        =%
        \cardin{\Disks \sqcap \pa}
    \end{math}
    is the number of disks in $\Disks$ that contain $\pa$.
\end{definition}

\begin{definition}
    A point $\pa$ is \emphi{$\eps$-safe}, with respect to an attack
    set $\BSet$, if
    \begin{math}
        \depthY{\pa}{\Disks - \BSet} \geq \eps
        \depthY{\pa}{\Disks}.
    \end{math}
    Namely, a point has at least an $\eps$-fraction of the disks
    originally covering it, even after the disks of the attack $\BSet$
    are removed.
\end{definition}

\begin{defn}
    \deflab{safe:zone}%
    Given a set of disks $\Disks$, the \emphw{arrangement} of
    $\Disks$, denoted by $\ArrX{\Disks}$ is the partition of the plane
    into faces, vertices and edges induced by $\Disks$, see
    \cite{bcko-cgaa-08}. A face/edge/vertex is thus
    \emphw{$\eps$-safe} if any point in it is $\eps$-safe. The union
    of all safe points forms the \emphi{$\eps$-safe zone}
    $\SZone = \SZoneY{\eps}{\Disks}{\BSet}$.

    Two points $\pa$ and $\pb$ in the plane are \emphi{$\eps$-safely
       connected}, if $\pa$ and $\pb$ belong to the same connected
    component of $\SZone$.
\end{defn}

\paragraph*{The problem}
The task at hand is to construct a sparse graph $\GS \subseteq \GD$
such that for any attack set $\BSet$, and any two $\eps$-safe points
$\pa,\pb$ that lie in the same connected component of
$\SZone = \SZoneY{\eps}{\Disks}{\BSet}$, there are two disks
$\disk_\pa, \disk_\pb \in \Disks - \BSet$, such that
$\pa \in \disk_\pa, \pb \in \disk_\pb$, and $\disk_\pa$ and
$\disk_\pb$ are connected in $\GS - \BSet$.

\subsection{Expander construction via random coloring}

\subsubsection{Expander Construction} %
\seclab{expander:construction}

Our purpose here is to build a sparse ``expander-like'' graph over a
set $\VV$ of $\n$ objects. Let $\eps \in (0,1)$ be a parameter.  Let
$\nc$ be a fixed number, such that $\n \leq \nc \leq 2\n$.  For some
sufficiently large constant $\cExpander>2$, consider the following
algorithm of generating a random graph.

\begin{description}
    \item The algorithm repeats the following
    $M = \cExpander \ceil{\eps^{-2}}$ times:

    It randomly colors the elements of $\VV$ with $\nc$ colors. For each such coloring, it connects two
    objects by an edge if their colors differ by $1$ (modulo $\nc$).
\end{description}
The final graph $\G$ (over $\VV$) results from using all the computed
edges in all these iterations.

\subsubsection{Proving expansion properties}

\begin{lemma}
    \lemlab{first}%
    Let $S$ be a set of $\n$ objects, and
    $\chi: S \rightarrow \IRX{\nc}$ be a random coloring of $S$, for
    $\nc \geq \n$.  Let
    $X = \cardin[1]{\chi(S)} = \cardin[1]{\Set{\chi(v)}{v \in S}}$ be
    the number of different colors used in $S$. Then, we have
    \begin{math}
        \Prob{X <  \n/e^2}
        <%
        \exp\pth{-\n}.
    \end{math}
\end{lemma}
\begin{proof}
    Let $T \subseteq \IRX{\nc}$ be a fixed set of $t = \beta \n$ colors. We let
    \begin{math}
        \gamma =%
        \Prob{\chi(S) \subseteq T} =%
        (t/\nc)^\n.
    \end{math}
    As such, we have that the probability that
    $\cardin{\chi(S)} \leq t$, is at most
    \begin{align*}
      &\sum_{T \subseteq \IRX{\nc}, \cardin{T} = t}%
        \Prob{\chi(S) \subseteq T}
        \leq%
        \binom{\nc}{t} \gamma
        \leq%
        e^t \pth{\frac{\nc}{t}}^t\pth{\frac{t}{\nc}}^{\n}
        =%
        e^t\pth{\frac{t}{\nc}}^{\n-t}
        \leq
        e^t\beta^{\n-t}
        \pth{\frac{\n}{\nc}}^{\n-t}.
    \end{align*}
    Since  $t \leq \beta \n$ and letting $\beta \leq 1/e^2$, we have
    \begin{math}
        e^t \beta^{\n-t} \leq \exp( t - 2(\n-t) ) \leq%
        \exp( -\n ).
  \end{math}
\end{proof}

For a coloring $\chi_i$, and an object $s \in \VV$, we denote the set
of all elements in $\VV$ that are connected to $s$ in this coloring by
\begin{equation*}
    \nbrY{i}{s}
    =%
    \Set{ t \in \VV \bigl.}{ |\chi_i(t) - \chi_i(s)| \equiv 1  \bmod \nc}.
\end{equation*}

\begin{defn}
    Let $\eps \in (0,1)$ be a parameter.  A graph $\G =(\VV, \EG)$
    with $\n$ vertices, is an \emphi{$\eps$-connector} if
    \begin{equation*}
        \forall S \subseteq \VV,
        \cardin{S} \geq \eps \n \implies
        \cardin{\nbrX{S}} > (1-\eps)\n.        
    \end{equation*}    
\end{defn}

\begin{lemma}
    \lemlab{eps_set_expansion}%
    Let $\cExpSize, \cExpRep$ be the two sufficiently large constants
    used in the above construction, and let $\G$ be the graph
    built. Then, for $\n \geq \cExpSize/\eps^2$, and
    $M \geq \cExpRep/\eps^2$, we have, with probability
    $\geq 1 - \exp(-4\n)$, that $\G$ is an $\eps/4$-connector.
\end{lemma}
\begin{proof}
    Fix the set $S$, and a ``bad'' set $T$ (disjoint from $S$) of size
    $(\eps/4) \n$. Here, the bad event is that $S$ is not connected
    to $T$ in $\G$. For some $s \in S$, the probability that $s$ is
    not adjacent to any vertex $t \in T$, by an edge induced by a
    specific coloring $\chi_i$ is
    \begin{equation*}
        \Prob{\nbrY{i}{s} \cap T = \varnothing} \leq%
        1 - \frac{|\chi_i(T)|}{\nc}.        
    \end{equation*}
    Conceptually, we first color the elements of $T$, and then color
    the elements of $S$.  The case $\cardin{S} > (1-\eps/4) \n$ is not
    possible as $S$ and $T$ are disjoint. Using the
    independence of the neighborhood of each vertex $s \in S$, for
    each coloring $i$, we have that
    $\Prob{\nbrY{i}{S} \cap T = \varnothing}$ is
    \begin{equation*}
        \prod_{s \in S} \Prob{\nbrY{i}{s} \cap T = \varnothing}
        \leq%
        \biggl({1 - \frac{\abs{\chi_i(T)}}{\nc}}\biggr)^{\cardin{S}}%
        \leq
        \exp \Bigl( - \frac{\eps \n}{4 \nc }  |\chi_i(T)| \Bigr)
        \leq
        \exp \Bigl( - \frac{  \eps}{8  } |\chi_i(T)| \Bigr).
    \end{equation*}
    A coloring $\chi_i$ is \emphw{good} if
    \begin{math}
        |\chi_i(T)| \geq {\cardin{T}}/{e^2} \geq \tau = \eps \n/40.
    \end{math}
    By \lemref{first}, we have that
    \begin{equation*}        
        \Prob{\chi_i \text{ is bad}\bigr.}
        =%
        \Prob{\cardin{\chi_i(T)} < \tau\bigr.}%
        \leq%
        \exp\pth{-|T|}
        =%
        \exp\Bigl( -\frac{\eps}{4} \n \Bigr). 
    \end{equation*}
    Thus, we have
    \begin{align*}
      \beta_i
      &=
        \Prob{\nbrY{i}{S} \cap T = \varnothing}
        \leq%
        \ProbCond{\nbrY{i}{S} \cap T = \varnothing}{
        \text{$\chi_i$ is good}}
        +
        \Prob{\chi_i \text{ is bad}\bigr.}
      \\
      &\leq%
        \exp \Bigl( - \frac{  \eps}{8  } \cdot
        \frac{\eps \n}{40} \Bigr)        
        +
        \exp\Bigl( -\frac{\eps}{4} \n \Bigr)
        \leq%
        2\exp \Bigl( - \frac{  \eps^2}{320  } \nu  \Bigr)        
        \leq%
        \exp \Bigl( - \frac{  \eps^2}{640  } \nu  \Bigr).
    \end{align*}
    for $\n > 640/\eps^2$.  As the colorings $\chi_1, \dots, \chi_M$
    are chosen independently, we have that
    \begin{equation*}
        \beta%
        =%
        \Prob{\nbrX{S} \cap T = \varnothing}%
        =%
        {\textstyle\prod}_{i = 1}^{M}\beta_i
        \leq%
        \exp \pth{- M\eps^2 \n/640}.
    \end{equation*}
    There are
    $\binom{\n}{\geq \eps \n/4} \binom{\n}{\eps \n/4} \leq 4^\n$
    choices for the sets $S$ and $T$. Thus, using the union bound over
    all of the choices for these sets, we have that
    \begin{equation*}
        \Prob{\exists S,T : \nbrX{S} \cap T = \varnothing} \leq%
        \sum_{S,T} \Prob{\nbrX{S} \cap T = \varnothing}%
        \leq%
        4^n \beta%
        \leq%
        \exp \biggl(2\n - \frac{M\eps^2 \n}{640}\biggr)%
        \leq%
        \exp \pth{-4\n},        
    \end{equation*}
    if $M \geq 2600/\eps^2$.
\end{proof}

\section{The construction of the safely connected %
   subgraph}

\subsection{Preliminaries}

For a set $\Disks$ of disks, and a parameter $k$, let
$\GIY{\leq k}{\Disks}$ be the subgraph of the intersection graph,
where two disks $\disk_1, \disk_2 \in \Disks$ are connected by an
edge, if there exists a point $\pa$, such that
$\pa \in \disk_1 \cap\disk_2$, and the depth of $\pa$ in $\Disks$ at
most $k$.  Such an edge is \emphi{$k$-shallow} in $\Disks$.
Similarly, let $\ArrY{\leq k}{\Disks}$ be the arrangement formed by
keeping only vertices, edges and faces of the arrangement
$\ArrX{\Disks}$ that are of depth at most $k$ (i.e., each connected
region of all points of deeper depth form a ``hole'' face in this
arrangement).

\subsubsection{Computing the shallow parts of the intersection graph}

The following is well known \cite{cs-arscg-89,by-ag-98}
-- for the sake of completeness, we provide a proof.

\begin{lemma}
    \lemlab{face_comb_complexity}%
    Let $\Disks$ be a set of $n$ disks, and let $k$ be a parameter.
    We have that the combinatorial complexity of
    $\ArrY{\leq k}{\Disks}$ is $O(nk)$, and this also bounds
    $\cardin[1]{\EdgesX{\GIY{\leq k}{\Disks}}}$.  Both the
    arrangement $\ArrY{\leq k}{\Disks}$ and the graph
    $\GIY{\leq k}{\Disks}$ can be computed in $O(n \log n + nk)$
    (expected) time.
\end{lemma}

\begin{proof}
    The first part is a well known consequence of the Clarkson-Shor
    technique \cite{cs-arscg-89}, as the union complexity of $n$ disks
    is linear. The construction algorithm for the arrangement is
    described by Boissonnat and Yvinec \cite{by-ag-98}.
	
    The second part (which is also known) follows by applying the
    Clarkson-Shor technique. Indeed, every face/vertex/edge of depth
    at most two of $\ArrX{\Disks}$, can contribute one edge to
    $\GIY{\leq 2}{\Disks}$. As such, the number of such edges is
    $O(n)$ as the complexity of $\ArrY{\leq 2}{\Disks}$ is $O(n)$.
	
    Let $\EG = \EGX{\GIY{\leq 2}{\Disks}}$.  Consider an edge
    $\edge = \disk_1, \disk_2 \in \EdgeSet$, with a point
    $\pa \in \edge$ being the witness point of depth at most $k$ such
    that $\pa \in \disk_1 \cap \disk_2$.  Let $R$ be a random sample
    of disks from $\Disks$, where each disk is sampled with
    probability $\alpha = 1/k$. The probability of this edge to appear
    in $\GIY{\leq 2}{R}$ is
    \begin{equation*}
        \Prob{\edge \in \EGX{\GIY{\leq 2}{R}}}
        \geq%
        (1/k)^2 (1-1/k)^{k-2}%
        \geq%
        1/(10k^2),        
    \end{equation*}
    Indeed, this is the probability of picking $\disk_1,\disk_2$ to
    the sample, and no other disks (of the at most $k-2$ disks) that
    covers $\pa$.
	
    The complexity of $\ArrY{\leq 2}{R}$ is bounded by $O(|R|)$, and
    as $\Ex{|R|} = O(n/k)$, it follows that
    \begin{equation*}
        \sum_{\edge \in \EdgeSet} \Prob{ \edge \in \EdgesX{\GIY{\leq
                 2}{R}}}
        \leq%
        O\pth{ \Ex{|{\ArrY{\leq 2}{R}} | \bigr.}}
        =%
        O\pth{ \Ex{|R | \bigr.}}
        =%
        O(n/k).
    \end{equation*}
    As for the other direction, we have
    \begin{equation*}
        \sum_{\edge \in \EdgeSet} \Prob{ \edge \in \EdgesX{\GIY{\leq
                 2}{R}}}
        \geq%
        \frac{\cardin{\EdgeSet}}{10 k^2}.
    \end{equation*}
    Combining the above two, we have
    $\cardin{\EdgeSet} / k^2 = O(n/k)$, which implies
    $\cardin{\EdgeSet} = O(nk)$.

    Having the arrangement $\ArrY{\leq k}{\Disks}$ is not by itself
    sufficient to compute efficiently the graph
    $\GIY{\leq k}{\Disks}$. Instead, one can lift the disks to planes,
    and use $n/k$-shallow cuttings \cite{ct-oda23-16}. This results in
    a decomposition of the plane into $O(n/k)$ cells, such that each
    cell has a conflict-list of size $O(k)$. We compute the
    arrangement of the disks in the conflict list, and by tracing the
    boundary of each disk, it is straightforward to discover all the
    edges of $\GIY{\leq k}{\Disks}$ that arise out of points in this
    cell. This takes $O(k^2)$ time per cell, and $O(n k + n \log n)$
    time overall, since computing the shallow cuttings takes
    $O(n \log n)$ time. This also provides an alternative algorithm
    for computing $\ArrY{\leq k}{\Disks}$.
\end{proof}

\subsubsection{The bipartite case}

Analogous to the previous section, for two sets of disks
$\Disks_1, \Disks_2$ and a parameter $k$, we let
$\GIY{\leq k}{\Disks_1, \Disks_2}$ be the intersection graph defined
as before, but where edges are only present between disks
$\disk_1, \disk_2$ such that $\disk_1 \in \Disks_1$ and
$\disk_2 \in \Disks_2$. Using the preceding lemma and this definition,
the following corollary is immediate.

\begin{corollary}
    \corlab{bipartite_comb_complexity}%
    Let $\DSX{1}$ and $\DSX{2}$ be two disjoint sets of disks of total
    size $n$.  We have that the number of edges of
    \begin{math}
        \G_{\leq k} (\DSX{1}, \DSX{2})
    \end{math}
    is bounded by $O\pth{ n k}$.  Additionally, the edges of this
    graph can be computed, in $O( n k )$ time, for
    $k = \Omega(\log n)$.
\end{corollary}

\subsection{The construction algorithm}
\seclab{construction}

The input is a set of $n$ disks $\Disks$, and a parameter
$\eps \in (0,1)$, where $\cExpSize, \cExpRep$ are the constants from
\lemref{eps_set_expansion}, and $\constA$ is a constant to be
specified later.  Initially, the algorithm starts with the empty graph
over $\Disks$. Let
\begin{equation}
    \alpha = \ceil{\constA \cExpSize \pth{\eps^{-2} + 4 \ln n}}.
    \eqlab{threshold}
\end{equation}
Next, using \lemref{face_comb_complexity}, the algorithm computes all
the faces of depth $\leq \alpha$ in $\ArrX{\Disks}$, and adds all the
edges induced by $\alpha$-shallow intersections of the
\lemref{face_comb_complexity} to the graph.

This takes care of all the shallow faces of the arrangement. For
deeper faces, in the $i$\th round, for
$i=1, \ldots, 1+ \ceil{\log_2 (n/\alpha) }$, the algorithm sets
$\alpha_i = 2^{i-1}\alpha$. The algorithm handles the faces with depth
in the range $(\alpha_{i-1}, \alpha_i]$, as follows:

\begin{description}
    \item \hypertarget{hook:i:round}{} \textbf{$i$\th round}:
    For $j = 1, \ldots, M = \ceil{\cExpRep/\eps^2}$, the algorithm colors
    each disk of $\Disks$ uniformly at random with $\alpha_i$
    colors. Let $\chi_{i,j} : \Disks \to \IRX{\alpha_i}$ be this
    coloring. Let $\DSX{t} = \chi_{i,j}^{-1} (t)$ be the set of disks
    of $\Disks$ colored by color $t$, for $t \in \IRX{\alpha_i}$. Using 
    \corref{bipartite_comb_complexity} as a subroutine, the algorithm computes the edges of
    $\EGX{ \G_{\leq \alpha} (\DSX{t-1}, \DSX{t})}$ and adds them to
    the resulting graph, for $t=1, \ldots, \alpha_i$ (where
    $\DSX{0}=\DSX{\alpha_i}$).
\end{description}
Namely, the algorithm computes $O( \log^2 n)$ random colorings in each
round, and adds the shallow edges between consecutively colored disks,
for each such coloring to the graph.

The final graph is denoted by $\spanner = \spX{\Disks}$.

\subsection{Analysis}

\subsubsection{Construction time and size}%

\begin{lemma}
    \lemlab{runtime}%
    The construction algorithm runs in time
    $O( n \eps^{-4} \log^2 n)$. This also bounds the number of edges
    in the computed graph.
\end{lemma}

\begin{proof}
    At the start, the algorithm includes all edges from faces with
    depth $\leq \alpha$. By \lemref{face_comb_complexity}, this can be
    done in $O(n \alpha) = O(n/\eps^2 + n \log n)$ time.
    
    For inner iteration $j$, the disk coloring step can be performed
    in $O(n)$ time by simply randomly assigning each disk a color from
    $\IRX{\alpha_i}$. Let $n_{t} = \abs{\DSX{t}}$, where $\DSX{t} \subseteq \Disks$
    is the set of disks assigned color $t$.  The algorithm computes
    edges induced by $\alpha$-shallow faces in
    $\Arr(\DSX{t-1} \cup \DSX{t} )$ for $t = 1, \dots, \alpha_i$.  By
    \corref{bipartite_comb_complexity}, these edges can be computed in
    $O(\alpha (n_{t} + n_{t+1}))$ time. Summing over $t$ (for a fixed
    $j$), we have that edges of this iteration can be computed in
    \[
        \sum_{t = 1}^{\alpha_i} O(\alpha (n_{t} + n_{t+1})) =%
        O(\alpha n)%
        =%
        O(n/\eps^2 + n \log n),
    \]
    time, as $\sum_{t = 1}^{\alpha_i} n_t = n$.  This is being
    performed $O( 1/\eps^2 )$ times in each round, and there are
    $O(\log n)$ rounds. Thus, the total work of this algorithm is
    $O( (n\eps^{-2} + n \log n) \eps^{-2} \log n) = O( n \eps^{-4}
    \log^2 n)$.
\end{proof}

\subsubsection{Rejecting edges from deep faces}

One issue that may arise is ignoring
faces which are too deep in the execution of the construction
algorithm. This can happen if there is some face in the
arrangement $\ArrX{\Disks}$ with depth in the range
$(\alpha_{i-1}, \alpha_i]$ that, in the $i$\th round, has more than
$\alpha$ disks intersecting it from a given color pair. So even under
the random coloring, the face still has depth which is too large under
some color pair, and some of its induced edges are ignored. Thus, we
must upper bound the probability of any failure of this type for any
color pair $(t-1,t)$, and any coloring $\chi_{i,j}$ sampled in the
$i$\th round.

\begin{lemma}
    \lemlab{faces_are_not_too_deep} Consider some face
    \(\face \in \ArrX{\Disks}\) such that
    $\depthX{\face} \in (\alpha_{i-1}, \alpha_i]$. In the $j$\th
    coloring of the $i$\th round of the algorithm (see
    \secref{construction}), for any fixed color $t \in \IRX{\alpha_i}$,
    the probability that \(\face\) is a hole (i.e., has depth bigger
    than $\alpha$) in $\ArrY{\leq \alpha}{\DSX{t-1}, \DSX{t}}$ is
    bounded by $1/n^{O(1)}$.
\end{lemma}

\begin{proof}
    The face $\face$ is only a hole in
    $\ArrY{\leq \alpha}{\DSX{t-1}, \DSX{t}}$ during the $i$\th round
    only if $\depthX{f, \DSX{t-1} \cup \DSX{t}} > \alpha$.  Since each
    disk is colored uniformly at random with $\alpha_i$ colors, the
    probability that some disk incident to $\face$ has color $t-1$ or
    $t$ is $2 / \alpha_i$. Since there are
    $\depthX{\face} \leq \alpha_i$ disks incident to $\face$, we can
    bound the probability of this event (taken over the choice of
    coloring $\chi_{i,j}$) by
    \begin{equation*}
        \Prob{\depthX{f, \DSX{t-1} \cup \DSX{t}} > \alpha}
        \leq%
        \binom{\depthX{\face}}{\alpha} \pth{\frac{2}{\alpha_i}}^{\alpha}
        \leq%
        \pth{\frac{\alpha_i e}{\alpha}}^\alpha
        \pth{\frac{2}{\alpha_i}}^{\alpha}
        =%
        \pth{\frac{2e}{\alpha}}^{\alpha}
        \leq%
        \exp( -\alpha )
        \leq%
        \frac{1}{n^{O(1)}}.
    \end{equation*}
    as 
    \begin{math}
        \pth{\frac{n}{k}}^k \leq \binom{n}{k} \leq \pth{\frac{e
              n}{k}}^k,
    \end{math}
    and by making
    $\constA$ sufficiently large.
\end{proof}

Now, we apply this lemma with a union bound to show that, with high probability,
we never ignore any edges needed for our construction due to having a face 
colored with too much of a single color.

\begin{corollary}
    \corlab{no_faces_ever_ignored}
    With probability $\geq 1 - n^{-7}$, no faces are ever ignored during the
    iteration they are handled in the spanner construction from \secref{construction}.
\end{corollary}

\begin{proof}
    To never have any face ignored, every face must have depth at most $\alpha$ under any
    consecutive color pair $t-1,t$, under all $M$ colorings, in the iteration $i$ it is handled.
    There are $\alpha_i \leq n$ possibilities for the color $t$, and by 
    \lemref{face_comb_complexity}, there are at most $O(n^2)$ distinct faces in $\ArrX{\Disks}$. 
    Thus, applying \lemref{faces_are_not_too_deep} and union bounding, the probability any face is 
    ever ignored during the iteration in which it is handled is at most
    \begin{align*}
      \sum_{\face \in \ArrX{\Disks}} \sum_{j = 1}^{M} \sum_{t \in \IRX{\alpha_i}}
      \Prob{\depthX{f, \DSX{t-1} \cup \DSX{t}} > \alpha}
      &\leq%
        O(n^2) \cdot \frac{1}{n^{O(1)}}
        =%
        n^{-7},
    \end{align*}
    which implies the desired lower bound on the probability of no faces ever being ignored.
\end{proof}

\subsubsection{The depth expander property}

For a point $\pa$ in the plane, the set $\Disks\sqcap \pa$ is the set
of all disks of $\Disks$ that contains $\pa$.  For the intersection
graph $\G = \IGX{\Disks}$, the induced subgraph $\G_{\Disks \sqcap \pa}$ is a
clique. We claim that the induced graph $\spanner_{\Disks \sqcap \pa}$
is an expander.

\begin{lemma}
    \lemlab{d:connector}%
    For any point in the plane $\pa$, let
    $\n = \depthY{\pa}{\Disks} = \cardin{\Disks \sqcap \pa}$ be its depth in
    the set of disks $\Disks$. We have that
    $\spanner_{\Disks \sqcap \pa}$ is an $\eps/4$-connector. This
    holds for all the points in the plane with probability
    $\geq 1-1/n^{10}$.
\end{lemma}
\begin{proof}
    If $\n \leq \alpha$, then $\spanner_{\Disks \sqcap \pa}$ is a
    clique, and the claim holds. Otherwise, fix $i$ such that
    $\n \leq \alpha_i < 2\n$. Observe that for the disks of
    $\Disks \sqcap \pa$, the \hyperlink{hook:i:round}{$i$\th round} of
    the algorithm, the construction is identical to the connector
    construction of \secref{expander:construction}. As such, the
    resulting graph (which might have more edges because of other
    iterations), has the properties of \lemref{eps_set_expansion},
    with probability $\geq 1-\exp(-4\n) \geq 1-n^{16}$, since
    $\nu > \alpha > 4 \ln n$. As the arrangement has at most $O(n^2)$
    faces, the claim follows by the union bound.
\end{proof}

\subsubsection{Safe connectivity}

\begin{lemma}
    \lemlab{safe_connectivity}
    For any attack set $\BSet \subseteq \Disks$, let
    $\SZone = \SZoneY{\eps}{\Disks}{\BSet}$ be the safe zone (see
    \defref{safe:zone}), and let $\pa, \pb$ be two points that are in
    the same connected component of $\SZone$. Then, there is a path in
    $\spanner - \BSet$, between two disks
    $\disk_\pa, \disk_\pb \in \Disks - \BSet$, where
    $\pa \in \disk_\pa$ and $\pb \in \disk_\pb$.

    This property holds for all attack sets, with probability
    $\geq 1-n^{4}$.
\end{lemma}
\begin{proof}
    Let $\CC$ be the connected component of $\SZone$ that contains
    $\pa$ and $\pb$. The arrangement $\ArrX{\Disks}$ restricted to
    $\CC$ is connected and has at most $O(n^2)$ faces, edges and
    vertices, and as such, there is a path $\gamma$ between $\pa$ and
    $\pb$, inside $\CC$, that crosses at most $O(n^2)$ faces and edges
    of $\Disks$. Indeed, we can assume without loss of generality that $\gamma$
    does not cross an edge of
    $\ArrX{\Disks}$ more than once, as otherwise it can be
    shortcut. Next, construct a sequence of points on $\gamma$, as
    follows. Initially, let $\pa_1 = \pa$. Then, continuously move
    along $\gamma$ towards $\pb$. Any time the traversal enters a new
    entity (i.e., face/edge/vertex) of $\ArrX{\Disks}$, the traversal
    places a new witness point on the new entity -- specifically, in
    the connected component of $\gamma$ formed with the intersection
    of this entity (in the interior of this intersection if possible)
    that contains that current location. The final point is
    $\pa_m = \pb$. By the above, we have that $m = O(n^2)$.
    Let $R(i)$ be the set of all the disks of $\Disks - \BSet$
    that are reachable from $\pa$ in $\spanner - \BSet$. Here, the
    start set is $R(1) = (\Disks - \BSet) \sqcap \pa$.

    For all $i$, let
    $\Leftover_{i} = (\Disks - \BSet ) \sqcap \pa_{i-1}$,
    $\nL_i = \cardin{\Leftover_i}$, and
    $d_i = \depthY{\pa_{i}}{\Disks}$. We claim (indecisively), that
    for any $i$, we have that
    $|R(i)| \geq (\eps/2) \depthY{\pa_i}{\Disks}$.

    The claim readily holds for $i=1$, as
    $|R(1)| \geq \eps \depthY{\pa}{\Disks}$, since $\pa$ is in the
    safe zone, and there are $\nL_i$ disks of $\Disks - \BSet$
    that cover $\pa$ and are thus reachable
    from $\pa$. So assume this holds for all $j <i$, and consider
    $R(i-1)$ and $R(i)$. Observe that they differ by at most two
    disks, under general position assumption.

    If $d_{i-1} < \alpha$ and $d_i < \alpha$, then
    $\spanner \sqcap \pa_{i-1}$ and $\spanner \sqcap \pa_{i}$ are
    cliques. By induction, $|R(i-1)| \geq (\eps/2) d_i$, which
    implies that all the disks of $\Leftover_i$ are reachable from
    $\pa$. By the general position assumption $R(i-1) \cap R(i)$ is not
    empty, which implies that there is at least one disk of $R(i)$
    that is reachable from $\pa$. Since
    $\spanner_{\Disks \sqcap \pa_{i-1}}$ is a clique, it follows that all
    the disks of $\Leftover_{i}$ are reachable from $\pa$, and since
    $\nL_i \geq \eps d_i$, the claim follows.
    
    If $d_{i-1} \geq \alpha$ and $d_i \geq \alpha$, then
    $|R(i-1)| \geq (\eps/2) d_{i-1} \geq (\eps/2)\alpha > 10/\eps$,
    see \Eqref{threshold}.  We have that at most two disks of $R(i-1)$
    are not present in $R(i)$, which implies that
    $|R(i)| \geq |R(i-1)|-2 \geq (\eps/2)d_{i-1} -2 \geq (\eps/4)d_i$,
    as $|d_i - d_{i-1}| \leq 2$. By the expansion property of
    \lemref{d:connector}, this implies that at least $(1-\eps/4)d_i$
    vertices of $R(i-1)$ are connected to $R(i)$ in
    $\spanner \sqcap \pa_i$. Let $Z_i$ be the of vertices in
    $\Disks \sqcap \pa_i$ that are not connected to $R(i-1)$, and
    observe that $|Z_i| \leq (\eps/4) d_i$.  As such, we have that
    \begin{equation*}
        |R(i)|
        \geq%
        \nL_i - |Z_i|
        =%
        \eps d_i - (\eps/4)d_i
        \geq (3/4)\eps d_i,
    \end{equation*}
    which implies the claim.
    
    The remaining cases, where $d_i, d_{i-1} \in [\alpha-2,\alpha+2]$
    are handled in similar fashion, and we omit the easy details.  We
    thus conclude that the point $\pa_i$, for all $i$, has at least
    $(\eps/2)d_i$ disks that are reachable to it from $\pa$ in the
    graph $\spanner - \BSet$.
    
    The expansion property of \lemref{d:connector} holds with probability
    \(\geq 1 - 1/n^{10}\), and so we obtain our final result by union
    bounding over all \(n^4\) possible pairs of adjacent faces.
\end{proof}

\subsection{The result}

\begin{theorem}
    Let $\Disks$ be a set of $n$ disks in the plane, $\eps \in (0,1)$
    be a parameter. The above algorithm constructs a sparse graph
    $\spanner$, which is a sparse subgraph of the intersection graph
    $\IGX{\Disks}$, such that for any attack set
    $\BSet \subseteq \Disks$, and any two points $\pa, \pb$ in the
    plane, such that $\pa$ and $\pb$ are
    \defrefY{safe:zone}{$\eps$-safely connected}, there is a path in
    $\spanner - \BSet$ between a disk that contains $\pa$ and
    a disk that contains $\pb$.
    
    This property holds for any attack set, and any two points, with
    probability $\geq 1-1/n^{3}$. The construction time and the number
    of edges of $\spanner$ is bounded by $O(\eps^{-4} n \log^2 n)$.
\end{theorem}

\begin{proof}
    Follows from the bounds on failure events in \lemref{safe_connectivity} and 
    \corref{no_faces_ever_ignored} and applying a union bound.
\end{proof}

\section{Conclusions}

We presented a new technique for sparsifying the intersection graph of
disks (or any shapes that their union complexity is near linear) --
the resulting graph has the property of preserving connectivity in the
regions that are still $\eps$-covered after an attack (potentially
involving a large fraction of the disks) are being deleted. There are
other guarantees that one might want, for example, the reliable
guarantee we have for spanners \cite{bho-sda-20} -- that is, that if
an attack deletes $b$ disks in the spanner, then deleting $(1+\eps)b$
disks in the original intersection graph would leave the original
graph with similar connected components provided by the original
graph. It is natural to conjecture that our construction (or a
slightly modified construction) provides such a guarantee. We leave
this as an open problem for further research.

\BibTexMode{%
   \SoCGVer{%
      \bibliographystyle{plainurl}%
   }%
   \NotSoCGVer{%
      \bibliographystyle{alpha}%
   }%
   \bibliography{reliable_disks}
}%
\BibLatexMode{\printbibliography}

\end{document}